\documentclass[a4paper]{article}
\usepackage[margin=30mm]{geometry}

\usepackage{graphicx}

\usepackage{amsmath,amssymb}
\usepackage{url}
\usepackage{rotating}
\usepackage{comment}
\usepackage{multirow}
\usepackage{amsthm}
\theoremstyle{plain}
\newtheorem{theorem}{Theorem}

\newtheorem{lemma}[theorem]{Lemma}
\newtheorem{corollary}[theorem]{Corollary}

\usepackage[algoruled,linesnumbered,algo2e,vlined]{algorithm2e}

\newcommand{\emptystr}{\varepsilon}

\newcommand{\lcp}{\mathsf{lcp}} 
\newcommand{\lcpinf}{\mathsf{lcp}^{\infty}}
\newcommand{\LCPinf}{\mathsf{LCP}^{\infty}} 
\newcommand{\lcpmax}{\lambda} 

\newcommand{\N}{\mathbf{N}}

\newcommand{\sSigma}{\Sigma_{s}}
\newcommand{\pSigma}{\Sigma_{p}}
\newcommand{\ssigma}{\sigma_{s}}
\newcommand{\psigma}{\sigma_{p}}
\newcommand{\penc}[1]{\langle#1\rangle}
\newcommand{\fce}{\ensuremath{\pi}} 

\newcommand{\SR}{\mathsf{R}}
\newcommand{\SUR}{\mathsf{R}^{-1}}

\newcommand{\LF}{\mathsf{LF}}
\newcommand{\FL}{\mathsf{FL}}
\newcommand{\Lstr}{\mathsf{L}}
\newcommand{\Fstr}{\mathsf{F}}
\newcommand{\rank}{\mathsf{rank}}
\newcommand{\select}{\mathsf{select}}

\newcommand{\GetInterval}{\mathsf{GetMI}}

\newcommand{\LFNQ}{\mathsf{FPQ}}
\newcommand{\RFNQ}{\mathsf{FNQ}}

\newcommand{\occ}{\mathsf{occ}}

\newcommand{\RmQ}{\mathsf{RmQ}}

\newcommand{\idtt}[1]{\ensuremath{\mathtt{#1}}}

\usepackage{hyperref}
\usepackage{xcolor}

\title{Breaking a Barrier in Constructing Compact Indexes for Parameterized Pattern Matching} 

\author{
  Kento~Iseri\\
  {Kyushu Institute of Technology, Fukuoka, Japan}\\
  {\texttt{iseri.kento210@mail.kyutech.jp}}\\
  \\
  Tomohiro~ I\\
  {Kyushu Institute of Technology, Fukuoka, Japan}\\
  {\texttt{tomohiro@ai.kyutech.ac.jp}}\\
  \\
  Diptarama~Hendrian\\
  {Tohoku University, Sendai, Japan}\\
  {\texttt{diptarama@tohoku.ac.jp}}\\
  \\
  Dominik~K{\"{o}}ppl\\
  {Department of Computer Science, Universität M\"{u}nster, Germany}\\
  {\texttt{dominik.koeppl@uni-muenster.de}}\\
  \\
  Ryo~Yoshinaka\\
  {Tohoku University, Sendai, Japan}\\
  {\texttt{ryoshinaka@tohoku.ac.jp}}\\
  \\
  Ayumi~Shinohara\\
  {Tohoku University, Sendai, Japan}\\
  {\texttt{ayumis@tohoku.ac.jp}}\\
}

\date{}
\begin{document}

\maketitle              
\begin{abstract}
A parameterized string (p-string) is a string over an alphabet $(\sSigma \cup \pSigma)$, 
where $\sSigma$ and $\pSigma$ are disjoint alphabets for static symbols (s-symbols) and for parameter symbols (p-symbols), respectively.
Two p-strings $x$ and $y$ are said to parameterized match (p-match) if and only if
$x$ can be transformed into $y$ by applying a bijection on $\pSigma$ to every occurrence of p-symbols in $x$.
The indexing problem for p-matching is to preprocess a p-string $T$ of length $n$ 
so that we can efficiently find the occurrences of substrings of $T$ that p-match with a given pattern.
Extending the Burrows-Wheeler Transform (BWT) based index for exact string pattern matching,
Ganguly et al. [SODA 2017] proposed the first compact index (named pBWT) for p-matching,
and posed an open problem on how to construct it in compact space, i.e., in $O(n \lg |\sSigma \cup \pSigma|)$ bits of space.
Hashimoto et al. [SPIRE 2022] partially solved this problem by showing how to construct
some components of pBWTs for $T$ in $O(n \frac{|\pSigma| \lg n}{\lg \lg n})$ time in an online manner
while reading the symbols of $T$ from right to left.
In this paper, we improve the time complexity to $O(n \frac{\lg |\pSigma| \lg n}{\lg \lg n})$.
We remark that removing the multiplicative factor of $|\pSigma|$ from the complexity is of great interest 
because it has not been achieved for over a decade
in the construction of related data structures like parameterized suffix arrays even in the offline setting.
We also show that our data structure can support backward search, a core procedure of BWT-based indexes, at any stage of the online construction,
making it the first compact index for p-matching that can be constructed in compact space and even in an online manner.
\end{abstract}

\section{Introduction}\label{sec:intro}
A \emph{parameterized string} (\emph{p-string}) is a string over an alphabet $(\sSigma \cup \pSigma)$, 
where $\sSigma$ and $\pSigma$ are disjoint alphabets for \emph{static symbols} (\emph{s-symbols}) and for \emph{parameter symbols} (\emph{p-symbols}), respectively.
Two p-strings $x$ and $y$ are said to \emph{parameterized match} (\emph{p-match}) if and only if
$x$ can be transformed into $y$ by applying a bijection on $\pSigma$ to every occurrence of p-symbols in $x$.
P-matching was introduced by Baker aiming at software maintenance and plagiarism detection
~\cite{1993Baker_TheorOfParamPatterMatch_STOC,1996Baker_ParamPatterMatchAlgorAnd,1997Baker_ParamDuplicInStrinAlgor},
and has been extensively studied in the last decades (see a recent survey~\cite{2020MendivelsoTP_BriefHistorOfParamMatch} and references therein).

The indexing problem for p-matching is to preprocess a p-string $T$ of length $n$ 
so that we can efficiently find the occurrences of substrings of $T$ that p-match with a given pattern.
Extending indexes for exact string pattern matching,
there have been proposed several data structures that can be used as indexes for p-matching, e.g.,
parameterized suffix trees~\cite{1993Baker_TheorOfParamPatterMatch_STOC,1995Kosaraju_FasterAlgorForConstOf,1996Baker_ParamPatterMatchAlgorAnd,1997Baker_ParamDuplicInStrinAlgor}, 
parameterized suffix arrays~\cite{2008DeguchiHBIT_ParamSuffixArrayForBinar_PSC,2009IDBIT_LightParamSuffixArrayConst_IWOCA,2012BealA_PSuffixSortinAsArith,2019FujisatoNIBT_DirecLinearTimeConstOf_SPIRE},
parameterized suffix trays~\cite{2021FujisatoNIBT_ParamSuffixTray_CIAC},
parameterized DAWGs~\cite{2022NakashimaFHNYIBST_ParamDawgsEfficConstAnd},
parameterized position heaps~\cite{2017DiptaramaKONS_PositHeapsForParamStrin_CPM,2018FujisatoNIBT_RightToLeftOnlinConst_PSC,2019FujisatoNIBT_ParamPositHeapOfTrie_CIAC} and
parameterized Burrows-Wheeler Transform (pBWT)~\cite{2017GangulyST_PbwtAchievSuccinDataStruc_SODA,2021KimC_SimplFmIndexForParam,2022GangulyST_FullyFunctParamSuffixTrees_ICALP}.

Among these indexes, pBWTs are the most space economic, consuming $n \lg \sigma + O(n)$ bits~\cite{2017GangulyST_PbwtAchievSuccinDataStruc_SODA} or
$2n \lg \sigma + 2n + o(n)$ bits with a simplified version proposed in~\cite{2021KimC_SimplFmIndexForParam},
where $\sigma$ is the alphabet size.
Let $\ssigma$ and respectively $\psigma$ be the numbers of distinct s-symbols and p-symbols that appear in $T$.
The pBWT of $T$ can be constructed via the parameterized suffix tree of $T$
for which $O(n (\lg \ssigma + \lg \psigma))$-time or randomized $O(n)$-time construction 
algorithms are known~\cite{1995Kosaraju_FasterAlgorForConstOf,2003ColeH_FasterSuffixTreeConstWith,2011LeeNP_LineConstOfParamSuffix},
but the intermediate memory footprint of $O(n \lg n)$ bits could be intolerable when $n$ is significantly larger than $\sigma$.
Hashimoto et al.~\cite{2022HashimotoHKYS_ComputParamBurrowWheelTrans_SPIRE} showed 
how to construct some components of the pBWT of~\cite{2021KimC_SimplFmIndexForParam} for $T$
in $O(n \frac{\psigma \lg n}{\lg \lg n})$ time in an online manner
while reading the symbols of $T$ from right to left.

In this paper, we improve the time complexity of~\cite{2022HashimotoHKYS_ComputParamBurrowWheelTrans_SPIRE} to $O(n \frac{\lg \psigma \lg n}{\lg \lg n})$.
Removing the multiplicative factor of $\psigma$ from the time complexity of~\cite{2022HashimotoHKYS_ComputParamBurrowWheelTrans_SPIRE} is of great interest 
because it has not been achieved for over a decade
in the construction of related data structures like parameterized suffix arrays even in an offline setting~\cite{2019FujisatoNIBT_DirecLinearTimeConstOf_SPIRE}.
We also show that our data structure can support backward search, a core procedure of BWT-based indexes, at any stage of the online construction,
making it the first compact index for p-matching that can be constructed in compact space and even in an online manner.
This is not likely to be achieved with the previous work~\cite{2022HashimotoHKYS_ComputParamBurrowWheelTrans_SPIRE}
due to the lack of support for 2D range counting queries in the data structure it uses.

Our computational assumptions are as follows:
\begin{itemize}
  \item We assume a standard Word-RAM model with word size $\Omega(\lg n)$.
  \item Each symbol in $(\sSigma \cup \pSigma)$ is represented by $O(\lg n)$ bits.
  \item We can distinguish symbols in $\sSigma$ and $\pSigma$ in $O(1)$ time, 
    e.g., by having some flag bits or thresholds separating them each other.
  \item The order on s-symbols are determined in $O(1)$ time based on their bit representations.
\end{itemize}

An index of a p-string $T$ for p-matching is to support, given a pattern $w$,
the counting queries that compute the number of occurrences of substrings that p-match with $w$ and
the locating queries that compute the occurrences of $w$.
Our main result is:
\begin{theorem}\label{theo:index}
  For a p-string $T$ of length $n$ over an alphabet $(\sSigma \cup \pSigma) \subseteq [0..\sigma]$,
  an index of $T$ for p-matching can be constructed online in
  $O(n \frac{\lg \psigma \lg n}{\lg \lg n})$ time and $O(n \lg \sigma)$ bits of space,
  where $\psigma$ is the number of distinct p-symbols used in the p-string.
  At any stage of the online construction, it can support the counting queries 
  in $O(m \frac{\lg \psigma \lg n}{\lg \lg n})$ time and 
  the locating queries in $O(m \frac{\lg \psigma \lg n}{\lg \lg n} + \occ \frac{\lg^2 n}{\lg \sigma \lg \lg n})$ time,
  where $m$ is the pattern length and $\occ$ is the number of occurrences to be reported.
\end{theorem}

\section{Preliminaries}\label{sec:prelim}
\subsection{Basic notations and tools}
An integer interval $\{ i, i+1, \dots, j\}$ is denoted by $[i..j]$, 
where $[i..j]$ represents the empty interval if $i > j$.

Let $\Sigma$ be an ordered finite \emph{alphabet}.
An element of $\Sigma^*$ is called a \emph{string} over $\Sigma$.
The length of a string $w$ is denoted by $|w|$. 
The empty string $\emptystr$ is the string of length 0,
that is, $|\emptystr| = 0$.
Let $\Sigma^+ = \Sigma^* - \{\emptystr\}$ and $\Sigma^k = \{ x \in \Sigma^* \mid |x| = k \}$ for any non-negative integer $k$.
The concatenation of two strings $x$ and $y$ is denoted by $x \cdot y$ or simply $xy$.
When a string $w$ is represented by the concatenation of strings $x$, $y$ and $z$ (i.e. $w = xyz$), 
then $x$, $y$ and $z$ are called a \emph{prefix}, \emph{substring}, and \emph{suffix} of $w$, respectively.
A substring $x$ of $w$ is called \emph{proper} if $x \neq w$.

The $i$-th symbol of a string $w$ is denoted by $w[i]$ for $1 \leq i \leq |w|$,
and the substring of a string $w$ that begins at position $i$ and
ends at position $j$ is denoted by $w[i..j]$ for $1 \leq i \leq j \leq |w|$,
i.e., $w[i..j] = w[i]w[i+1] \cdots w[j]$.
For convenience, let $w[i..j] = \emptystr$ if $j < i$; further let $w[..i] = w[1..i]$ and $w[i..] = w[i..|w|]$ denote abbreviations for the prefix of length $i$ and the suffix starting at position $i$, respectively.
For two strings $x$ and $y$, let $\lcp(x, y)$ denote the length of the longest common prefix between $x$ and $y$.
We consider the lexicographic order over $\Sigma^*$ by extending the strict total order $<$ defined on $\Sigma$:
$x$ is lexicographically smaller than $y$ (denoted as $x < y$) if and only if 
either $x$ is a proper prefix of $y$ or $x[\lcp(x, y)+1] < y[\lcp(x, y)+1]$ holds.

For any string $w$, character $c$, and position $i~(1 \le i \le |w|)$,
$\rank_c(w, i)$ returns the number of occurrences of $c$ in $w[..i]$
and $\select_c(w, i)$ returns the $i$-th occurrence of $c$ in $w$.
For $1 \le i \le j \le |w|$, a \emph{range minimum query} $\RmQ_{w}(i, j)$ asks for $\arg \min_{i \le k \le j} \{ w[k] \}$.
We also consider \emph{find previous/next queries} $\LFNQ_{p}(w, i)$ and $\RFNQ_{p}(w, i)$,
where $p$ is a predicate either in the form of ``$c$'' (equal to $c$), ``$< c$'' (less than $c$) or ``$\ge c$'' (larger than or equal to $c$):
$\LFNQ_{p}(w, i)$ returns the largest position $j \le i$ at which $w[j]$ satisfies the predicate $p$.
Symmetrically, $\RFNQ_{p}(w, i)$ returns the smallest position $j \ge i$ at which $w[j]$ satisfies the predicate $p$.
For example with the integer string $w = [ 2, 5, 10, 6, 8, 3, 14, 5]$, 
$\RFNQ_{5}(w, 4) = 8$, $\RFNQ_{6}(w, 4) = 4$, $\LFNQ_{5}(w, 4) = 2$, $\RFNQ_{<5}(w, 4) = 6$, $\LFNQ_{< 5}(w, 4) = 1$, $\RFNQ_{\ge 9}(w, 4) = 7$ and $\LFNQ_{\ge 9}(w, 4) = 3$.

If the answer of $\select_c(w, i)$, $\LFNQ_{p}(w, i)$ or $\RFNQ_{p}(w, i)$ does not exist, it is just ignored.
To handle this case of non-existence, we would use them in an expression with $\min$ or $\max$:
For example, $\max\{1, \LFNQ_{p}(w, i) \}$ returns $1$ if $\LFNQ_{p}(w, i)$ does not exist.

Dynamic strings should support insertion/deletion of a symbol to/from any position as well as fast random access.
We use the following result:
\begin{lemma}[\cite{2015MunroN_ComprDataStrucForDynam_ESA}]\label{lem:drs}
  A dynamic string of length $n$ over an alphabet $[0..\sigma]$
  can be implemented while supporting random access, insertion, deletion, $\rank$ and $\select$ queries
  in $(n + o(n)) \lg \sigma$ bits of space 
  and $O(\frac{\lg n}{\lg \lg n})$ query and update times.
\end{lemma}

Dynamic binary strings equipped with $\rank$ and $\select$ queries can be used as
a building block for the dynamic wavelet matrix~\cite{2015ClaudeNP_WavelMatrixEfficWavelTree} of a string over an alphabet $[0..\sigma]$
to support queries beyond $\rank$ and $\select$.
The idea is that each of the other queries can be simulated by performing one of the building block queries on every level of the wavelet matrix, 
which has $\lceil \lg \sigma \rceil$ levels, cf.~\cite[Section 6.2.]{navarro14wavelet}.
\begin{lemma}\label{lem:dwm}
  A dynamic string of length $n$ over an alphabet $[0..\sigma]$ with $\sigma = O(n)$
  can be implemented while supporting random access, insertion, deletion, 
  $\rank$, $\select$, $\RmQ$, $\LFNQ$ and $\RFNQ$ queries
  in $(n + o(n)) \lceil \lg \sigma \rceil + O(\lg \sigma \lg n)$ bits of space
  and $O(\frac{\lg \sigma \lg n}{\lg \lg n})$ query and update times.
\end{lemma}

\subsection{Parameterized strings}
Let $\sSigma$ and $\pSigma$ denote two disjoint sets of symbols.
We call a symbol in $\sSigma$ a \emph{static symbol} (\emph{s-symbol}) 
and a symbol in $\pSigma$ a \emph{parameter symbol} (\emph{p-symbol}).
A \emph{parameterized string} (\emph{p-string}) is a string over $(\sSigma \cup \pSigma)$.
Let $\infty$ represent a symbol that is larger than any integer, and let $\N_{\infty} = \N_{+} \cup \{ \infty \}$ be the set of natural positive numbers $\N_{+}$ including infinity.
We assume that $\N_{\infty} \cap \sSigma = \emptyset$ and $(\N_{\infty} \cup \sSigma)$ is an ordered alphabet
such that all s-symbols are smaller than any element in $\N_{\infty}$.
Let $\$$ be the smallest s-symbol, which will be used as an end-marker of p-strings.
For any p-string $w$ the \emph{p-encoded string} $\penc{w}$ of $w$, 
also proposed as $\mathtt{prev}_{\infty}(w)$ in~\cite{2021KimC_SimplFmIndexForParam}, 
is the string in $(\N_{\infty} \cup \sSigma)^{|w|}$ such that
\begin{equation*}
  \penc{w}[i] =
  \begin{cases}
    w[i]     & \mbox{if $w[i] \in \sSigma$}, \\
    \infty   & \mbox{if $w[i] \in \pSigma$ and $w[i]$ does not appear in $w[..i-1]$}, \\
    i - j    & \mbox{otherwise,}
  \end{cases}
\end{equation*}
where $j$ is the largest position in $[1..i-1]$ with $w[i] = w[j]$.
In other words, p-encoding transforms an occurrence of a p-symbol into 
the distance to the previous occurrence of the same p-symbol, or $\infty$ 
if it is the leftmost occurrence.
By definition, p-encoding is prefix-consistent, i.e., 
$\penc{w} = \penc{wc}[..|w|]$ for any symbol $c \in (\sSigma \cup \pSigma)$.
On the other hand, $\penc{w}$ and $\penc{cw}[2..]$ differ 
if and only if $c \in \pSigma$ occurs in $w$.
If it is the case, the leftmost occurrence $h$ of $c$ in $w$
is the unique position such that $\penc{w}$ and $\penc{cw}[2..]$ differ 
with $\penc{w}[h] = \infty$ and $(\penc{cw}[2..])[h] = \penc{cw}[h+1] = h$,
i.e., $h = \select_{c}(w, 1)$ and $h+1 = \select_{c}(cw, 2)$.

For any p-string $w$, let $|w|_{p}$ denote the number of distinct p-symbols in $w$,
i.e., $|w|_{p} = \rank_{\infty}(\penc{w}, |w|)$.
We define a function $\fce$ that maps a p-string $w$ over $(\sSigma \cup \pSigma)$
to an element in $(\sSigma \cup [1..|w|_{p}])$ such that
\begin{equation*}
  \fce(w) =
  \begin{cases}
    w[1]     & \mbox{if $w[1] \in \sSigma$}, \\
    |w[..h+1]|_{p} & \mbox{otherwise,}
  \end{cases}
\end{equation*}
where $h+1 = \min \{ |w|, \select_{w[1]}(w, 2) \}$.
In the second case, $\fce(w)$ represents the rank of p-symbol $w[1]$ when ordering all p-symbols 
by their leftmost occurrences in $w[2..]$, considering the rank of p-symbols not in $w[2..]$ to be $|w|_{p}$.
If $\select_{w[1]}(w, 2)$ exists, it holds that $h = \select_{\infty}(\penc{w[2..]}, \fce(w))$.
For convenience, let $\fce(\emptystr) = \$$.

For two p-strings $u$ and $v$, $\lcpinf(\penc{u}, \penc{v})$ denotes the number of 
$\infty$'s in the longest common prefix of $\penc{u}$ and $\penc{v}$.

The following lemma states basic but important properties of the p-string encoding and $\fce$,
which immediately follow from the definition (see Fig.~\ref{fig:penc_order} for illustrations).
\begin{lemma}\label{lem:penc_order}
  For any p-strings $x$ and $y$ over $(\sSigma \cup \pSigma)$ with
  $\lcpmax = \lcp(\penc{x[2..]}, \penc{y[2..]})$, $e = \lcpinf(\penc{x[2..]}, \penc{y[2..]})$ and $\penc{x[2..]} < \penc{y[2..]}$,
  Table~\ref{table:cases} shows a complete list of cases for $\lcp(\penc{x}, \penc{y})$, $\lcpinf(\penc{x}, \penc{y})$ and 
  the lexicographic order between $\penc{x}$ and $\penc{y}$, 
  where a case starting with A and B is under the condition that at least one of $\fce(x)$ and $\fce(y)$ is in $\sSigma$,
  and respectively none of $\fce(x)$ and $\fce(y)$ is in $\sSigma$.
\end{lemma}

\begin{table}[t]
  \begin{center}
    \caption{All cases considered in Lemma~\ref{lem:penc_order}.
  On the one hand, a case starting with letter~A assumes that at least one of $\fce(x)$ and $\fce(y)$ is in $\sSigma$, while on the other hand, a case starting with letter~B
  assumes that none of $\fce(x)$ and $\fce(y)$ is in $\sSigma$.
    Here, $h = \select_{x[1]}(x[2..], 1)$ and $h' = \select_{y[1]}(y[2..], 1)$.}
    \label{table:cases}
    \begin{tabular}{l|l|l|l|l}
      cases & additional conditions                   & $\lcp(\penc{x}, \penc{y})$ & $\lcpinf(\penc{x}, \penc{y})$ & lexicographic order \\
      \hline
      (A1)  & $\fce(x) \neq \fce(y)$                  & $0$                        & $0$                           & $\penc{x} < \penc{y}$ iff $\fce(x) < \fce(y)$ \\
      (A2)  & $\fce(x) = \fce(y)$                     & $\lcpmax+1$                & $e$                           & $\penc{x} < \penc{y}$ \\
      \hline
      (B1)  & $\fce(x) = \fce(y) \le e$               & $\lcpmax+1$                & $e$                           & $\penc{x} < \penc{y}$ \\
      (B2)  & $\fce(x) \le e$ and $\fce(x) < \fce(y)$ & $h$                        & $\fce(x)$                     & $\penc{x} < \penc{y}$ \\
      (B3)  & $\fce(y) \le e$ and $\fce(y) < \fce(x)$ & $h'$                       & $\fce(y)$                     & $\penc{y} < \penc{x}$ \\
      (B4)  & $e < \min\{\fce(x), \fce(y)\}$          & $\lcpmax+1$                & $e+1$                         & $\penc{x} < \penc{y}$ \\
      \hline
    \end{tabular}
  \end{center}
\end{table}

\begin{figure}[t]
  \center{%
    \includegraphics[scale=0.25]{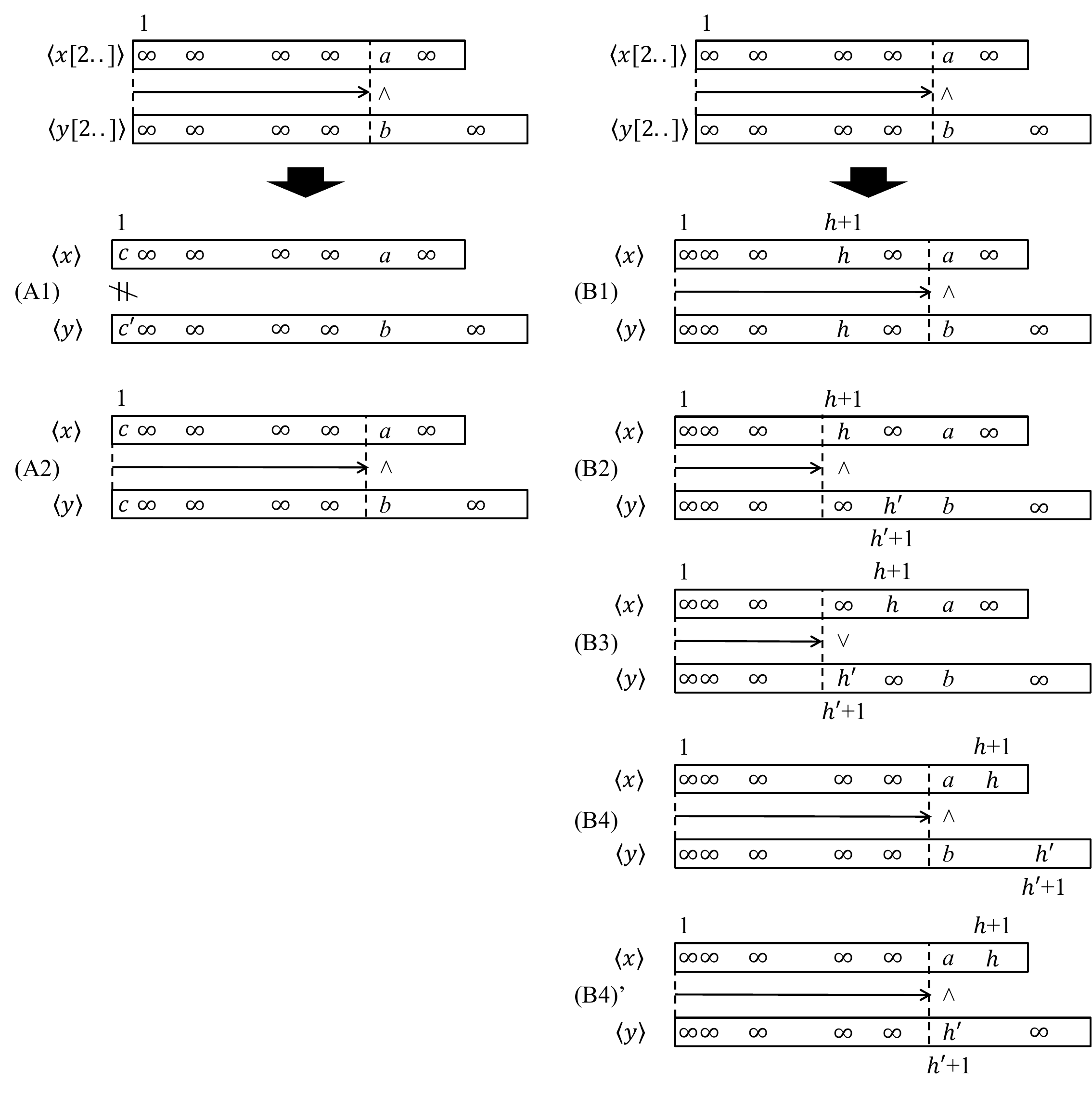}
  }
  \caption{Illustrations for the cases of Lemma~\ref{lem:penc_order}.
    Each right arrow represents the longest common prefix of two p-encoded strings, 
    and the lexicographic order between them is determined by the following p-encoded symbols.
    For Case (B1), $h = \select_{x[1]}(x[2..], 1) = \select_{y[1]}(y[2..], 1)$.
    For Case (B2)-(B4) and (B4)', $h = \select_{x[1]}(x[2..], 1)$ and $h' = \select_{y[1]}(y[2..], 1)$.
    Case (B4)' illustrates the case with $b = \infty$, which is included in Case (B4).
  }
  \label{fig:penc_order}
\end{figure}

By Lemma~\ref{lem:penc_order}, we have the following corollaries:
\begin{corollary}\label{cor:lcpinf_atmost}
  For any p-strings $x$ and $y$,
  $\lcpinf(\penc{x}, \penc{y}) \le \lcpinf(\penc{x[2..]}, \penc{y[2..]}) + 1$.
\end{corollary}

\begin{corollary}\label{cor:order_same}
  For any p-strings $x$ and $y$ with $\fce(x) = \fce(y)$,
  $x < y$ if and only if $\penc{x} < \penc{y}$.
\end{corollary}

\begin{corollary}\label{cor:order_e}
  For any p-strings $x$ and $y$ with $\fce(x) \le \lcpinf(\penc{x[2..]}, \penc{y[2..]}) < \fce(y)$, 
  it holds that $\penc{x}[..h+1] < \penc{x} < \penc{x}[..h'+1] < \penc{y}$, where
  $h = \select_{x[1]}(x[2..], 1)$ and $h' = \select_{y[1]}(y[2..], 1)$.
\end{corollary}

\begin{table}[h]
  \begin{center}
    \caption{An example of $\SUR_{T}(i)$, $\LCPinf_{T}$, $\Lstr_{T}$ and $\Fstr_{T}$ for a p-string $T=\idtt{xyazyxazxza\$}$ 
    with $\sSigma = \{ \idtt{a} \}$ and $\pSigma = \{\idtt{x}, \idtt{y}, \idtt{z} \}$.}
    \label{table:arrays}
    \begin{tabular}{|c||l|l||c|c|c|c|l|}
    \hline
    \multicolumn{1}{|c||}{$i$} & \multicolumn{1}{c|}{$T[i..]$} & \multicolumn{1}{c||}{$\penc{T[i..]}$} & \multicolumn{1}{c|}{$\SUR_{T}(i)$} & \multicolumn{1}{c|}{$\LCPinf_{T}[i]$}   & \multicolumn{1}{c|}{$\Lstr_T[i]$} & \multicolumn{1}{c|}{$\Fstr_T[i]$}  & \multicolumn{1}{c|}{$\penc{T[\SUR_{T}(i)..]}$} \\ \hline
  1    & $\idtt{xyazyxazxza\$}$    & $\idtt{\infty\infty a\infty35a432a\$}$    & 12    & 0    & $\idtt{a}$     & $\idtt{\$}$    & $\$$                                     \\
  2    & $\idtt{yazyxazxza\$}$     & $\idtt{\infty a\infty3\infty a432a\$}$    & 11    & 0    & 1              & $\idtt{a}$     & $\idtt{a\$}$                             \\
  3    & $\idtt{azyxazxza\$}$      & $\idtt{a\infty\infty\infty a432a\$}$      & 7     & 0    & 2              & $\idtt{a}$     & $\idtt{a\infty\infty2a\$}$               \\
  4    & $\idtt{zyxazxza\$}$       & $\idtt{\infty\infty\infty a432a\$}$       & 3     & 2    & 2              & $\idtt{a}$     & $\idtt{a\infty\infty\infty a432a\$}$     \\
  5    & $\idtt{yxazxza\$}$        & $\idtt{\infty\infty a\infty32a\$}$        & 10    & 0    & 2              & 1              & $\idtt{\infty a\$}$                      \\
  6    & $\idtt{xazxza\$}$         & $\idtt{\infty a\infty32a\$}$              & 6     & 1    & 3              & 2              & $\idtt{\infty a\infty32a\$}$             \\
  7    & $\idtt{azxza\$}$          & $\idtt{a\infty\infty2a\$}$                & 2     & 2    & 3              & 2              & $\idtt{\infty a\infty3\infty a432a\$}$   \\
  8    & $\idtt{zxza\$}$           & $\idtt{\infty\infty2a\$}$                 & 9     & 1    & 2              & 2              & $\idtt{\infty\infty a\$}$                \\  
  9    & $\idtt{xza\$}$            & $\idtt{\infty\infty a\$}$                 & 5     & 2    & 3              & 3              & $\idtt{\infty\infty a\infty32a\$}$       \\
  10   & $\idtt{za\$}$             & $\idtt{\infty a\$}$                       & 1     & 3    & $\idtt{\$}$    & 3              & $\idtt{\infty\infty a\infty35a432a\$}$   \\
  11   & $\idtt{a\$}$              & $\idtt{a\$}$                              & 8     & 2    & $\idtt{a}$     & 2              & $\idtt{\infty\infty2a\$}$                \\
  12   & $\idtt{\$}$               & $\idtt{\$}$                               & 4     & 2    & $\idtt{a}$     & 3              & $\idtt{\infty\infty\infty a432a\$}$      \\
    \hline
    \end{tabular}
  \end{center}
\end{table}

Let $T$ be a p-string that has the smallest s-symbol $\$$ as its end-marker, i.e., 
$T[|T|] = \$$ and $\$$ does not appear anywhere else in $T$.
The suffix rank function $\SR_{T} : [1..|T|] \rightarrow [1..|T|]$ for $T$ maps
a position $i~(1 \le i \le |T|)$ to the lexicographic rank of $\penc{T[i..]}$ in $\{\penc{T[j..]} \mid 1 \le j \le |T| \}$.
Its inverse function $\SUR_{T}(i)$ returns the starting position of the lexicographically $i$-th p-encoded suffix of $T$.

The main components of \emph{parameterized Burrows-Wheeler Transform (pBWT)} of $T$ are $\Fstr_{T}$ and $\Lstr_{T}$.
$\Fstr_{T}$ (resp. $\Lstr_{T}$) is defined to be
the string of length $|T|$ such that $\Fstr_{T}[i] = \fce(T[\SUR_{T}(i)..])$ 
(resp. $\Lstr_{T}[i] = \fce(T[\SUR_{T}(i)-1..])$), where we assume that $T[0..] = \$$.
\footnote{Previous studies~\cite{2017GangulyST_PbwtAchievSuccinDataStruc_SODA,2021KimC_SimplFmIndexForParam,2022HashimotoHKYS_ComputParamBurrowWheelTrans_SPIRE} 
define pBWTs based on sorted cyclic rotations, but our suffix-based definition is 
more suitable for online construction to prevent unnecessary update on $\Fstr_{T}$ and $\Lstr_{T}$.}
Since $\{T[\SUR_{T}(i)..] \mid 1 \le i \le |T| \} = \{T[\SUR_{T}(i)-1..] \mid 1 \le i \le |T| \}$
is equivalent to the set of all non-empty suffixes of $T$, $\Fstr_{T}$ is a permutation of $\Lstr_{T}$.

The so-called LF-mapping $\LF_{T}$ maps a position $i$ to $\SR_{T}(\SUR_{T}(i) - 1)$ if $\SUR_{T}(i) > 1$,
and otherwise $\SR_{T}(|T|) = 1$.
By definition and Corollary~\ref{cor:order_same}, we have:
\begin{corollary}\label{cor:pbwt_order}
  For any p-string $T$ and any integers $i, j$ with $1 \le i < j \le |T|$,
  $\LF_{T}(i) < \LF_{T}(j)$ if $\Lstr_{T}[i] = \Lstr_{T}[j]$.
\end{corollary}
Thanks to Corollary~\ref{cor:pbwt_order}, it holds that 
$\LF_{T}(i) = \select_{c}(\Fstr_{T}, \rank_{c}(\Lstr_{T}, i))$, where $c = \Lstr_{T}[i]$.
The inverse function $\FL_{T}$ of $\LF_{T}$ can be computed by 
$\FL_{T}(i) = \select_{c}(\Lstr_{T}, \rank_{c}(\Fstr_{T}, i))$, where $c = \Fstr_{T}[i]$.

Let $\LCPinf_{T}$ be the string of length $|T|$ such that
$\LCPinf_{T}[0] = 0$ and $\LCPinf_{T}[i] = \lcpinf(T[\SUR_{T}(i-1)..], T[\SUR_{T}(i)..])$ for any $1 < i \le |T|$.
An example of all explained arrays is given in Table~\ref{table:arrays}.

\section{Online construction algorithm}\label{sec:algorithm}
For online construction of our index for p-matching, we consider maintaining dynamic data structures for 
$\Fstr_{T}$, $\Lstr_{T}$ and $\LCPinf_{T}$ while prepending a symbol to the current p-string $T$.
The details of the data structures will be presented in Subsection~\ref{sec:analysis}.
In what follows, we focus on a single step of updating $T$ to $\hat{T} = cT$ for some symbol $c$ in $\sSigma \cup \pSigma$.
Note that $\Fstr_{T}$, $\Lstr_{T}$ and $\LCPinf_{T}$ are strongly related to the sorted p-encoded suffixes of a p-string
and $\hat{T} = cT$ is the only suffix that was not in the suffixes of $T$.
Let $k = \SR_{T}(1)$ and $\hat{k} = \SR_{\hat{T}}(1)$.
In order to deal with the new emerging suffix $\hat{T}$, 
we compute the lexicographic rank $\hat{k}$ of $\penc{\hat{T}}$ in the non-empty p-encoded suffixes of $\hat{T}$.
Then $\Fstr_{\hat{T}}$ and $\Lstr_{\hat{T}}$ can be obtained by 
replacing $\$$ in $\Lstr_{T}$ at $k$ by $\fce(\hat{T})$
and inserting $\$$ and $\fce(\hat{T})$ into the $\hat{k}$-th position of $\Lstr_{T}$ and $\Fstr_{T}$, respectively.
In Subsection~\ref{sec:khat}, we propose our algorithm to compute $\hat{k}$.
For updating $\LCPinf$, we have to compute the $\lcpinf$-values for $\penc{\hat{T}}$ 
with its lexicographically adjacent p-encoded suffixes,
which will be treated in Subsection~\ref{sec:lcpinf}.

\subsection{How to compute $\hat{k}$}\label{sec:khat}
Unlike the existing work~\cite{2022HashimotoHKYS_ComputParamBurrowWheelTrans_SPIRE} 
that computes $\hat{k}$ by counting the number of 
p-encoded suffixes that are lexicographically smaller than $\penc{\hat{T}}$, 
we get $\hat{k}$ indirectly by computing the rank of 
a lexicographically closest (smaller or larger) p-encoded suffix to $\penc{\hat{T}}$.
The lexicographically smaller (resp. larger) closest element in $\{ \penc{T[i..]} \mid 1 \le i \le |T| \}$ to $\penc{\hat{T}}$
is called the \emph{p-pred} (resp. \emph{p-succ}) of $\penc{\hat{T}}$.
and its rank is denoted by $k_{-}$ (resp. $k_{+}$).
Then it holds that $\hat{k} = k_{+} = k_{-} + 1$.

We start with the easy case that the prepended symbol $c$ is an s-symbol.
\begin{lemma}\label{lem:add_s-symbol}
  Let $\hat{T} = cT$ be a p-string with $c \in \sSigma$.
  If $p := \LFNQ_{c}(\Lstr_{T}, k)$ exists, 
  the rank $k_{-}$ of the p-pred of $\hat{T}$ is $\LF_{T}(p)$.
  Otherwise, $k_{-} = \select_{b}(\Fstr_{T}, \rank_{b}(\Fstr_{T}, |T|))$, where
  $b$ is the largest s-symbol that appears in $T$ and smaller than $c$.
\end{lemma}
\begin{proof}
  By Case (A2) of Lemma~\ref{lem:penc_order} (cf.~Table~\ref{table:cases}),  
  the lexicographic order of p-encoded suffixes starting with $c$
  does not change by removing their first characters, which are all $c$.
  If $p$ exists, 
  $\penc{T[\SUR_{T}(p)..]}$ is the lexicographically smaller closest p-encoded suffix
  to $\penc{T}$ that is preceded by $c$.
  Hence, $\penc{T[\SUR_{T}(\LF_{T}(p))..]} = \penc{c T[\SUR_{T}(p)..]}$
  is the p-pred of $\penc{cT} = \penc{\hat{T}}$, which means that
  $k_{-} = \LF_{T}(p)$.

  If $p$ does not exist, it implies that $\penc{\hat{T}}$ is the 
  lexicographically smallest p-encoded suffix that starts with $c$.
  Since $\penc{\hat{T}}$ lexicographically comes right after 
  the p-encoded suffixes starting with an s-symbol smaller than $c$,
  $k_{-}$ is the last occurrence of $b$ in $\Fstr_{T}$,
  that is, $k_{-} = \select_{b}(\Fstr_{T}, \rank_{b}(\Fstr_{T}, |T|))$.
\end{proof}

In the rest of this subsection, we consider the case that $c$ is a p-symbol.
If $T$ contains no p-symbol, it is clear that $k_{-} = |T|$.
Hence, in what follows, we assume that there is a p-symbol in $T$.

Since $\penc{\hat{T}}$ has the longest $\lcp$-value $\hat{\lcpmax}$ with its p-pred or p-succ,
we search for such p-encoded suffixes of $T$ using the following lemmas to leverage the information stored in $\LCPinf_{T}$.

\begin{lemma}\label{lem:lcp_rmq}
  Given two positions $i$ and $j$ with $1 \le i < j \le |T|$,
  $\lcpinf(\penc{T[\SUR_{T}(i)..]}, \penc{T[\SUR_{T}(j)..]}) = \RmQ_{\LCPinf_{T}}(i+1, j)$.
\end{lemma}
\begin{proof}
  It is not difficult to see that $\lcp(x, z) = \min\{\lcp(x, y), \lcp(y, z)\}$
  for any strings $x < y < z$, and thus, $\lcpinf(x, z) = \min\{\lcpinf(x, y), \lcpinf(y, z)\}$.
  Since $\LCPinf$ holds the $\lcpinf$-values of lexicographically adjacent p-encoded suffixes,
  we get $\lcpinf(\penc{T[\SUR_{T}(i)..]}, \penc{T[\SUR_{T}(j)..]}) = \min \{ \LCPinf_{T}[g] \}_{g = i+1}^{j} = \RmQ_{\LCPinf_{T}}(i+1, j)$
  by applying the previous argument successively.
\end{proof}

\begin{lemma}\label{lem:getinterval}
  Algorithm~\ref{algo:getinterval} correctly computes the maximal interval $[l..r]$ such that
  $\lcpinf(\penc{T[\SUR_{T}(i)..]}, \penc{T[\SUR_{T}(j)..]}) \ge e$ for any $j \in [l..r]$.
\end{lemma}

\begin{algorithm2e}[t]
\caption{Algorithm to compute the maximal interval $[l..r]$ such that
$\lcpinf(\penc{T[\SUR_{T}(i)..]}, \penc{T[\SUR_{T}(j)..]}) \ge e$ for any $j \in [l..r]$.}
\label{algo:getinterval}
\SetKw{Return}{return}
\SetKwFunction{GetInterval}{GetMI}
\SetKwProg{Fn}{Function}{:}{}
\Fn{\GetInterval{$i$, $e$}}
{
  $l \leftarrow \max\{1, \LFNQ_{< e}(\LCPinf_{T}, i)\}$\;
  $r \leftarrow \min\{|T|, \RFNQ_{< e}(\LCPinf_{T}, i+1) - 1\}$\;
  \Return $[l..r]$\;
}
\end{algorithm2e}

\begin{lemma}\label{lem:hatk}
  Algorithm~\ref{algo:hatk} correctly returns $\hat{k}$.
\end{lemma}
\begin{proof}
  Let $h_{i} = \select_{\infty}(\penc{T}, i)$ for any $1 \le i \le \min\{ |T|_{p}, \fce(\hat{T})\}$, 
  and $h_{i} = |T|+1$ for any $i > \min\{ |T|_{p}, \fce(\hat{T})\}$.
  Also let $\lcpmax = \max \{ \lcp(\penc{\hat{T}}, \penc{T[i..]}) \mid 1 \le i \le |T| \}$.
  Although Algorithm~\ref{algo:hatk} does not intend to compute the exact value of $\lcpmax$, 
  it checks if $\lcpmax$ falls in $[h_{e}..h_{e+1}]$ in decreasing order of $e$
  starting from $\min\{ \fce(\hat{T}), \max\{\LCPinf_{T}[k], \LCPinf_{T}[k+1]\} \}$.
  One of the necessary conditions to have $\lcp(\penc{\hat{T}}, \penc{T[i..]}) > h_{e}$
  is that $\lcp(\penc{T}, \penc{T[i+1..]}) \ge h_{e}$, or equivalently $\lcpinf(\penc{T}, \penc{T[i+1..]}) \ge e$.
  Line~\ref{hatk:lr} computes the maximal interval $[l..r]$ that represents the ranks of the p-encoded suffixes 
  having an $\lcpinf$-value larger than or equal to $e$.
  The basic idea is to find a p-encoded suffix in $\{ \penc{T[\SUR_{T}(p)..]} \}_{p = l}^{r}$
  that comes closest to $\penc{\hat{T}}$ when extended by adding its preceding symbol.
  When $e$ comes down to the point with $\lcpmax \in [h_{e}+1..h_{e+1}]$, 
  $\hat{k}$ is returned in one of the if-then-blocks at 
  Lines~\ref{hatk:fce_pred},~\ref{hatk:fce_succ},~\ref{hatk:e_pred} and~\ref{hatk:e_succ}.

  If $\lcp(\penc{\hat{T}}, \penc{T[i..]}) = h_{\hat{e}}$ for an integer $\hat{e}$, there are two possible scenarios:
  \begin{enumerate}
    \item $\lcpinf(\penc{T}, \penc{T[i+1..]}) \ge \hat{e}$ and either $\fce(\hat{T})$ or $\fce(T[i..])$ is $\hat{e}$, and
    \item $\lcp(\penc{T}, \penc{T[i+1..]}) = h_{\hat{e}} - 1$ and both $\fce(\hat{T})$ and $\fce(T[i..])$ are at least $\hat{e}$.
  \end{enumerate}
  The former p-encoded suffix is processed in one of the if-then-blocks at Lines~\ref{hatk:fce_at} and~\ref{hatk:e_at} when $e = \hat{e}$,
  while the latter at Lines~\ref{hatk:fce_pred},~\ref{hatk:fce_succ},~\ref{hatk:e_pred} and~\ref{hatk:e_succ}
  when $e = \hat{e} - 1$.
  Note that the former is never farther from $\penc{\hat{T}}$ than the latter because 
  the lexicographic order between $\penc{\hat{T}}$ and $\penc{T[i..]}$ is determined by $\infty$ and $h_{\hat{e}}$ at $h_{\hat{e}}+1$
  in the former case, while it is by $\infty$ and something smaller than $h_{\hat{e}}$ in the latter case.
  Since Algorithm~\ref{algo:hatk} processes the former case first, it guarantees that the algorithm finds the closer one first.

  The case with $e = \fce(\hat{T})$ is treated differently than other cases in the if-then-block at Line~\ref{hatk:fce}
  since $h_{\fce(\hat{T})}$ is the unique position where $\penc{T}[h_{\fce(\hat{T})}] = \infty$ turns into $\penc{\hat{T}}[h_{\fce(\hat{T})}+1] = h_{\fce(\hat{T})}$.
  For a p-encoded suffix $\penc{T[\SUR_{T}(q')..]} \in \{ \penc{T[\SUR_{T}(p)..]} \}_{p = l}^{r}$,
  having $\Lstr_{T}[q'] = \fce(\hat{T})$ is necessary and sufficient 
  for its extended suffix $\penc{T[\SUR_{T}(q')-1..]}$ to have an $\lcp$-value larger than $h_{\fce(\hat{T})}$ with $\hat{T}$.
  By Corollary~\ref{cor:order_same}, p-encoded suffixes satisfying this condition must preserve their lexicographic order after extension,
  and hence, it is enough to search for the closest one ($q \leftarrow \LFNQ_{e}(\Lstr_{T}, k)$ or $q \leftarrow \RFNQ_{e}(\Lstr_{T}, k)$) to $\penc{T}$ 
  and compute the rank of its extended suffix by $\LF_{T}(q)$.
  If Lines~\ref{hatk:fce_pred} and~\ref{hatk:fce_succ} fail to return a value, it means that $\lcpmax \le h_{\fce(\hat{T})}$.
  The if-block at Line~\ref{hatk:fce_at} checks if there exists a p-encoded suffix $\penc{T[i+1..]}$ that satisfies the former condition 
  to be $\lcp(\penc{\hat{T}}, \penc{T[i..]}) = h_{\fce(\hat{T})}$.
  It is enough to find one $\penc{T[\SUR_{T}(q)..]}$ with $\Lstr[q] \ge \fce(\hat{T})$ for $q \in [l..r]$
  because it is necessary and sufficient to have $\lcp(\penc{\hat{T}}, \penc{T[\SUR_{T}(q)-1..]}) = h_{\fce(\hat{T})}$ 
  and $\penc{\hat{T}}[h_{\fce(\hat{T})} + 1] = \infty \neq  h_{\fce(\hat{T})} = \penc{T[\SUR_{T}(q)-1..]}[h_{\fce(\hat{T})}+1]$.
  Note that there could be two or more p-encoded suffixes that satisfy the condition and their lexicographic order may change by extension.
  In the then-block at Line~\ref{hatk:fce_at}, the algorithm computes the rank of the 
  lexicographically smallest p-encoded suffix that has an $\lcpinf$-value larger than $\fce(\hat{T})$ with 
  $\penc{T[\SUR_{T}(q)-1..]} = \penc{T[\SUR_{T}(\LF_{T}(q))..]}$, which is the p-succ of $\hat{T}$ in this case.

  The case with $e \neq \fce(\hat{T})$ is processed in the else-block at Line~\ref{hatk:e}.
  Here, it is good to keep in mind that when we enter this else-block, 
  $\lcpinf(\penc{T}, \penc{T[i..]}) \le e$ or $\fce(T[i-1..]) \le e$ holds for any proper suffix $T[i..]$ of $T$,
  since otherwise $\hat{k}$ must be reporeted in a previous round of the foreach loop.

  When the if-condition at Line~\ref{hatk:e_pred} holds, $\penc{T[\SUR_{T}(q)..]}$ is the lexicographically 
  smaller closest p-encoded suffix to $\penc{T}$ such that $\lcpinf(\penc{\hat{T}}, \penc{T[\SUR_{T}(q)-1..]}) \ge e+1$,
  or equivalently $\lcp(\penc{\hat{T}}, \penc{T[\SUR_{T}(q)-1..]}) > h_{e}$.
  For any p-encoded suffix in $\{ \penc{T[\SUR_{T}(p)..]} \}_{p = q+1}^{k-1}$
  its extended suffix is lexicographically smaller than $\penc{T[\SUR_{T}(q)-1..]}$ due to Corollary~\ref{cor:order_e},
  and never closer to $\penc{\hat{T}}$ than $\penc{T[\SUR_{T}(q)-1..]}$.
  At Line~\ref{hatk:e_pred_mi}, the algorithm computes the maximal interval $[l'..r']$ such that every p-encoded suffix in 
  $\{ \penc{T[\SUR_{T}(p)..]} \}_{p = l'}^{r'}$ shares the common prefix of length
  $h' := \min \{ |T| - \SUR_{T}(q) + 1, \select_{\infty}(\penc{T[\SUR_{T}(q)..]}, e+1) \}$ with $\penc{T[\SUR_{T}(q)..]}$.
  Since any $\penc{T[i..]} \in \{ \penc{T[\SUR_{T}(p)..]} \}_{p = 1}^{l'-1}$ 
  has an $\lcp$-value smaller than $h'$ with $\penc{T[\SUR_{T}(q)..]}$, 
  it follows from Lemma~\ref{lem:penc_order} that $\penc{T[i-1..]} < \penc{T[\SUR_{T}(q)-1..]}$.
  Also, for any $\penc{T[i..]} \in \{ \penc{T[\SUR_{T}(p)..]} \}_{p = k+1}^{|T|}$,
  the aforementioned precondition to enter the else-block at Line~\ref{hatk:e} leads to 
  $\penc{\hat{T}} < \penc{T[i-1..]}$ or $\penc{T[i-1..]} < \penc{T[\SUR_{T}(q)..]}$ by Lemma~\ref{lem:penc_order}.
  So far we have confirmed that the $\lcp$-value between $\penc{\hat{T}}$ and the p-pred of $\penc{\hat{T}}$ is less than $h'$,
  which implies that the p-pred is the largest p-encoded suffix that is prefixed by $x := \penc{T[\SUR_{T}(q)-1..]}[..h']$.
  If $q' \leftarrow \LFNQ_{\ge e+2}(\Lstr_{T}, r')$ computed at Line~\ref{hatk:e_pred_qp} is in $[l'..r']$,
  $\penc{T[\SUR_{T}(q')-1..]} = \penc{T[\LF_{T}(q')..]}$ is prefixed by $x \cdot \infty$
  and the p-pred is the largest p-encoded suffix that is prefixed by $x \cdot \infty$,
  which can be computed by $\max \GetInterval(\LF_{T}(q'), e+2)$ because 
  $\penc{T[\SUR_{T}(q')-1..]}[..h'+1] = \penc{T[\SUR_{T}(\LF_{T}(q')..]}[..h'+1] = x \cdot \infty$ contains exactly $e+2$ $\infty$'s.
  If $q' \notin [l'..r']$, the p-pred is the largest p-encoded suffix that is prefixed by $x \cdot h'$,
  which is $\penc{T[\SUR_{T}(\LF_{T}(q))..]}$.

  When the if-condition at Line~\ref{hatk:e_succ} holds, $\penc{T[\SUR_{T}(q)..]}$ is the lexicographically 
  larger closest p-encoded suffix to $\penc{T}$ such that $\lcpinf(\penc{\hat{T}}, \penc{T[\SUR_{T}(q)-1..]}) \ge e+1$,
  or equivalently $\lcp(\penc{\hat{T}}, \penc{T[\SUR_{T}(q)-1..]}) > h_{e}$.
  For any p-encoded suffix in $\{ \penc{T[\SUR_{T}(p)..]} \}_{p = k+1}^{q-1}$
  its extended suffix is lexicographically smaller than $\penc{T[\SUR_{T}(q)-1..]}$ due to Corollary~\ref{cor:order_e},
  and never closer to $\penc{\hat{T}}$ than $\penc{T[\SUR_{T}(q)-1..]}$.
  At Line~\ref{hatk:e_succ_mi}, the algorithm computes the maximal interval $[l'..r']$ such that every p-encoded suffix in 
  $\{ \penc{T[\SUR_{T}(p)..]} \}_{p = l'}^{r'}$ shares the common prefix of length
  $h' := \min \{ |T| - \SUR_{T}(q) + 1, \select_{\infty}(\penc{T[\SUR_{T}(q)..]}, e+1) \}$ with $\penc{T[\SUR_{T}(q)..]}$.
  Since any $\penc{T[i..]} \in \{ \penc{T[\SUR_{T}(p)..]} \}_{p = r'+1}^{|T|}$ 
  has an $\lcp$-value smaller than $h'$ with $\penc{T[\SUR_{T}(q)..]}$, 
  it follows from Lemma~\ref{lem:penc_order} that $\penc{T[i-1..]} < \penc{T[\SUR_{T}(q)-1..]}$.
  Also, for any $\penc{T[i..]} \in \{ \penc{T[\SUR_{T}(p)..]} \}_{p = 1}^{k-1}$,
  the aforementioned precondition to enter the else-block at Line~\ref{hatk:e} leads to
  $\penc{T[i-1..]} < \penc{\hat{T}}$ by Lemma~\ref{lem:penc_order}.
  So far we have confirmed that the $\lcp$-value between $\penc{\hat{T}}$ and the p-succ of $\penc{\hat{T}}$ is less than $h'$,
  which implies that the p-succ is the smallest p-encoded suffix that is prefixed by $x := \penc{T[\SUR_{T}(q)-1..]}[..h']$.
  If $q' \leftarrow \RFNQ_{e+1}(\Lstr_{T}, l')$ computed at Line~\ref{hatk:e_pred_qp} is in $[l'..r']$, 
  the p-succ is the smallest p-encoded suffix that is prefixed by $x \cdot h'$,
  which is $\penc{T[\SUR_{T}(\LF_{T}(q'))..]}$.
  If $q' \notin [l'..r']$, the p-succ is the smallest p-encoded suffix that is prefixed by $x \cdot \infty$,
  which can be computed by $\min \GetInterval(\LF_{T}(q), e+2)$ because 
  $\penc{T[\SUR_{T}(q)-1..]}[..h'+1] = \penc{T[\SUR_{T}(\LF_{T}(q))..]}[..h'+1] = x \cdot \infty$ contains exactly $e+2$ $\infty$'s.

  When we enter the if-then-block at Line~\ref{hatk:e_at}, it is guaranteed that $\lcpmax \le h_{e}$.
  In order to check if there exists a p-encoded suffix $\penc{T[i+1..]}$ that satisfies the former condition 
  to be $\lcp(\penc{\hat{T}}, \penc{T[i..]}) = h_{e}$, 
  the algorithm computes $q \leftarrow \LFNQ_{e}(\Lstr_{T}, r))$.
  If $q \in [l..r]$, $\penc{T[\SUR_{T}(q)..]}$ is the lexicographically largest p-encoded suffix that satisfies the condition, 
  and by Corollary~\ref{cor:order_same}, its extended suffix $\penc{T[\SUR_{T}(q)-1..]}$ must be the largest one to have 
  $\lcp(\penc{\hat{T}}, \penc{T[\SUR_{T}(q)-1..]}) = h_{e}$.
  Therefore, $\penc{T[\SUR_{T}(q)-1..]}$ is the p-pred of $\penc{\hat{T}}$, and $\hat{k} = 1 + \LF_{T}(q)$.
\end{proof}

\begin{algorithm2e}[t]
\caption{Algorithm to compute $\hat{k}$.}
\label{algo:hatk}
\SetKw{DownTo}{down to}
\SetKw{Return}{return}
\SetKwFunction{GetInterval}{GetMI}
\ForEach{$e \leftarrow \min\{ \fce(\hat{T}), \max\{\LCPinf_{T}[k], \LCPinf_{T}[k+1]\} \}$ \DownTo $1$}{
  $[l..r] \leftarrow$ \GetInterval{$k$, $e$}\;\label{hatk:lr}
  \If{$e = \fce(\hat{T})$}{\label{hatk:fce}
    \lIf{$(q \leftarrow \LFNQ_{e}(\Lstr_{T}, k)) \in [l..r]$}{\label{hatk:fce_pred}
      \Return $1 + \LF_{T}(q)$
    }
    \lIf{$(q \leftarrow \RFNQ_{e}(\Lstr_{T}, k)) \in [l..r]$}{\label{hatk:fce_succ}
      \Return $\LF_{T}(q)$
    }
    \lIf{$(q \leftarrow \RFNQ_{\ge e+1}(\Lstr_{T}, l)) \in [l..r]$}{\label{hatk:fce_at}
      \Return $\min$ \GetInterval{$\LF_{T}(q)$, $e+1$}
    }
  }
  \Else{\label{hatk:e}
    \If{$(q \leftarrow \LFNQ_{\ge e+1}(\Lstr_{T}, k)) \in [l..r]$}{\label{hatk:e_pred}
      $[l'..r'] \leftarrow$ \GetInterval{$q$, $e+1$}\;\label{hatk:e_pred_mi}
      $q' \leftarrow \LFNQ_{\ge e+2}(\Lstr_{T}, r')$\;\label{hatk:e_pred_qp}
      \lIf{$q' \in [l'..r']$}{
        \Return $1 + \max$ \GetInterval{$\LF_{T}(q')$, $e+2$}
      }
      \lElse{
        \Return $1 + \LF_{T}(q)$
      }
    }
    \If{$(q \leftarrow \RFNQ_{\ge e+1}(\Lstr_{T}, k)) \in [l..r]$}{\label{hatk:e_succ}
      $[l'..r'] \leftarrow$ \GetInterval{$q$, $e+1$}\;\label{hatk:e_succ_mi}
      $q' \leftarrow \RFNQ_{e+1}(\Lstr_{T}, l')$\;\label{hatk:e_succ_qp}
      \lIf{$q' \in [l'..r']$}{
        \Return $\LF_{T}(q')$
      }
      \lElse{
        \Return $\min$ \GetInterval{$\LF_{T}(q)$, $e+2$}
      }
    }
    \lIf{$(q \leftarrow \LFNQ_{e}(\Lstr_{T}, r)) \in [l..r]$}{\label{hatk:e_at}
      \Return $1 + \LF_{T}(q)$
    }
  }
}
\end{algorithm2e}

\subsection{How to maintain $\LCPinf$}\label{sec:lcpinf}
Suppose that we have $k = \SR_{T}(1)$, $\hat{k} = \SR_{\hat{T}}(1)$, $\Lstr_{T}$, $\Fstr_{T}$, 
we show how to compute $\lcpinf$-values of $\penc{\hat{T}}$ with 
its p-pred $\penc{T[\SUR_{T}(\hat{k})-1..]}$ and p-succ $\penc{T[\SUR_{T}(\hat{k})..]}$ 
to maintain $\LCPinf$.

We focus on $\lcpinf(\penc{\hat{T}}, \penc{T[\SUR_{T}(\hat{k})..]})$ because the other one can be treated similarly.
We apply Lemma~\ref{lem:penc_order} by setting $x = \hat{T}$ and $y = T[\SUR_{T}(\hat{k})..]$
if $k < \FL_{T}(\hat{k})$ (otherwise we swap their roles for $x$ and $y$).
In order to get $\lcpinf(\penc{x}, \penc{y})$, all we need are 
$\fce(x) = \fce(\hat{T})$, $\fce(y) = \Fstr[\hat{k}]$ and
$e = \lcpinf(\penc{x[2..]}, \penc{y[2..]})$.
For the computation of $e$ we use Lemma~\ref{lem:lcp_rmq}, i.e,
$e = \lcpinf(\penc{x[2..]}, \penc{y[2..]}) = \RmQ_{\LCPinf_{T}}(k+1, \FL_{T}(\hat{k}))$.

\subsection{Dynamic data structures and analysis}\label{sec:analysis}
We consider constructing $\Fstr_{T}$, $\Lstr_{T}$ and $\LCPinf_{T}$ 
of a p-string $T$ over an alphabet $(\sSigma \cup \pSigma) \subseteq [0..\sigma]$ of length $n$ online.
Let $\ssigma$ and respectively $\psigma$ be the numbers of distinct s-symbols and p-symbols that appear in $T$.

First we show data structures needed to implement our algorithm presented in the previous subsections.
For $\Fstr_{T}$ we maintain a dynamic string of Lemma~\ref{lem:drs}
supporting random access, insertion, $\rank$ and $\select$ queries in $O(\frac{\lg n}{\lg \lg n})$ time 
and $O(n \lg \sigma)$ bits of space.
For $\LCPinf_{T}$ we maintain a dynamic string of Lemma~\ref{lem:dwm} to support
random access, insertion, $\RmQ$, $\LFNQ$ and $\RFNQ$ queries
in $O(\frac{\lg \psigma \lg n}{\lg \lg n})$ time
and $O(n \lg \psigma)$ bits of space.

If we build a dynamic string of Lemma~\ref{lem:dwm} for $\Lstr_{T}$,
the query time would be $O(\frac{\lg \sigma \lg n}{\lg \lg n})$.
Since our algorithm does not use $\RmQ$, $\LFNQ$ and $\RFNQ$ queries for s-symbols,
we can reduce the query time to $O(\frac{\lg \psigma \lg n}{\lg \lg n})$ as follows.
We represent $\Lstr_{T}$ with one level of a wavelet tree, where a bit vector $B_T$ partitions the alphabet into $\sSigma$ and $[1..|T|_{p}]$
and thus has pointers to $X_{T}$ and $Y_{T}$ storing respectively the sequence over $\sSigma$ and that over $[1..|T|_{p}]$ of $\Lstr_{T}$.
We represent the former and the latter by the data structures described in Lemmas~\ref{lem:drs} and~\ref{lem:dwm}, respectively, 
since we only need the aforementioned queries such as $\RmQ$ on $Y_{T}$.
Then, queries on $\Lstr_{T}$ can be answered in $O(\frac{\lg \psigma \lg n}{\lg \lg n})$ time
using $O(n \lg \sigma)$ bits of space.

In addition to these dynamic strings for $\Fstr_{T}$, $\Lstr_{T}$ and $\LCPinf_{T}$,
we consider another dynamic string $Z_{T}$.
Let $Z_{T}$ be a string that is obtained by extracting the leftmost occurrence of every p-symbol in $T$.
Thus, $|Z_{T}| = \psigma \le n$.
A dynamic string $Z_{T}$ of Lemma~\ref{lem:drs} enables us to compute $\fce(cT)$ for a p-symbol $c$ 
by $\fce(cT) = \min\{\infty, \select_{c}(Z_{T}, 1)\}$ in 
$O(\frac{\lg \psigma}{\lg \lg \psigma})$ time and $O(\psigma \lg \psigma)$ bits of space.

We also maintain a fusion tree of~\cite{Patrascu2014DynamicIntegerSets_OptimalRankSelectPred} 
of $O(\ssigma \lg n) = O(n \lg \sigma)$ bits to maintain the set of s-symbols used in the current $T$,
which enables us to compute $b$ in Lemma~\ref{lem:add_s-symbol} in $O(\frac{\lg \ssigma}{\lg \lg \ssigma})$ time.

We are now ready to prove the following lemma.
\begin{lemma}\label{lemma:online}
  $\Fstr_{T}$, $\Lstr_{T}$ and $\LCPinf_{T}$ for a p-string of length $n$ over an alphabet $(\sSigma \cup \pSigma) \subseteq [0..\sigma]$
  can be constructed online in
  $O(n \frac{\lg \psigma \lg n}{\lg \lg n})$ time and $O(n \lg \sigma)$ bits of space,
  where $\psigma$ is the number of distinct p-symbols used in the p-string.
\end{lemma}
\begin{proof}
  We maintain the dynamic data structures of $O(n \lg \sigma)$ bits 
  described in this subsection while prepending a symbol to the current p-string.
  For a single step of updating $T$ to $\hat{T} = cT$ with $c \in (\sSigma \cup \pSigma)$, 
  we compute $\hat{k} = \SR_{\hat{T}}(1)$ as described in Subsection~\ref{sec:khat} and
  obtain $\Fstr_{\hat{T}}$ and $\Lstr_{\hat{T}}$ by replacing $\$$ in $\Lstr_{T}$ at $k = \SR_{T}(1)$ by $\fce(\hat{T})$
  and inserting $\$$ and $\fce(\hat{T})$ into the $\hat{k}$-th position of $\Lstr_{T}$ and $\Fstr_{T}$, respectively.
  $\LCPinf$ is updated as described in Subsection~\ref{sec:lcpinf}.

  If $c \in \sSigma$, the computation of $\hat{k}$ based on Lemma~\ref{lem:add_s-symbol} requires a constant number of queries.
  If $c \in \pSigma$, Algorithm~\ref{algo:hatk} computes $\hat{k}$ invoking $O(2+e-\hat{e})$ queries, 
  where $e = \max\{\LCPinf_{T}[k], \LCPinf_{T}[k+1]\}$ and $\hat{e} = \max\{\LCPinf_{\hat{T}}[\hat{k}], \LCPinf_{\hat{T}}[\hat{k}+1]\}$.
  The value $e$ can be seen as a potential held by the current string $T$, which upper bounds the number of queries.
  The number of queries in a single step can be $O(\psigma)$ in the worst case 
  when $e$ and $\hat{e}$ are close to $\psigma$ and respectively $0$,
  but this will reduce the potential for later steps, which allows us to give an amortized analysis.
  Since a single step increases the potential at most 1 by Corollary~\ref{cor:lcpinf_atmost},
  the total number of queries can be bounded by $O(n)$.

  Since we invoke $O(n)$ queries that take $O(\frac{\lg \psigma \lg n}{\lg \lg n})$ time each,
  the overall time complexity is $O(n\frac{\lg \psigma \lg n}{\lg \lg n})$.
\end{proof}

\section{Extendable compact index for p-matching}\label{sec:bws}
In this section, we show that $\Lstr_{T}$, $\Fstr_{T}$ and $\LCPinf_{T}$ can serve as an index for p-matching.

First we show that we can support backward search,
a core procedure of BWT-based indexes, 
with the data structures for $\Lstr_{T}$, $\Fstr_{T}$ and $\LCPinf_{T}$
described in Subsection~\ref{sec:analysis}.
For any p-string $w$, let $w$-interval be the maximal interval $[l..r]$ such that 
$\penc{T[\SUR_{T}(p)..]}$ is prefixed by $\penc{w}$ for any $p \in [l..r]$.
We show the next lemma for a single step of backward search, which computes $cw$-interval from $w$-interval.
\begin{lemma}\label{lemma:bws}
  Suppose that we have data structures for $\Lstr_{T}$, $\Fstr_{T}$ and $\penc{\hat{T}}$ described in Subsection~\ref{sec:analysis}.
  Given $w$-interval $[l..r]$ and $c \in (\sSigma \cup \pSigma)$,
  we can compute $cw$-interval $[l'..r']$ in $O(\frac{\lg \psigma \lg n}{\lg \lg n})$ time.
\end{lemma}
\begin{proof}
  We show that we can compute $cw$-interval from $w$-interval
  using a constant number of queries supported on $\Lstr_{T}$, $\Fstr_{T}$ and $\penc{\hat{T}}$,
  which takes $O(\frac{\lg \psigma \lg n}{\lg \lg n})$ time each.
\begin{itemize}
  \item When $c$ is in $\sSigma$: 
    $\penc{T[\SUR_{T}(\LF_{T}(p))..]}$ is prefixed by $\penc{cw}$
    if and only if $\penc{T[\SUR_{T}(p)..]}$ is prefixed by $\penc{w}$ and $\Lstr_{T}[p] = c$.
    In other words, $\LF_{T}(p) \in [l'..r']$ if and only if $p \in [l..r]$ and $\Lstr_{T}[p] = c$.
    Then it holds that $l' = \LF_{T}(\RFNQ_{c}(\Lstr_{T}, l))$ and $r' = \LF_{T}(\LFNQ_{c}(\Lstr_{T}, r))$
    due to Corollary~\ref{cor:order_same}.
  \item When $c$ is a p-symbol that appears in $w$:
    Similar to the previous case, $\LF_{T}(p) \in [l'..r']$ if and only if $p \in [l..r]$ and $\Lstr_{T}[p] = \fce(cw)$.
    Then it holds that $l' = \LF_{T}(\RFNQ_{\fce(cw)}(\Lstr_{T}, l))$ and $r' = \LF_{T}(\LFNQ_{\fce(cw)}(\Lstr_{T}, r))$
    due to Corollary~\ref{cor:order_same}.
  \item When $c$ is a p-symbol that does not appear in $w$:
    Let $e = |w|_{p}$.
    Since $p \in [l..r]$ and $\Lstr_{T}[p] > e$ are necessary and sufficient conditions for $\LF_{T}(p)$ to be in $[l'..r']$,
    we can compute $r' - l' + 1$, the width of $[l'..r']$, by counting the number of positions $p$ such that $\Lstr_{T}[p] > e$ with $p \in [l..r]$.
    This can be done with 2D range counting queries, which can also be supported with the wavelet tree of Lemma~\ref{lem:dwm}.
    If $s = \LF_{T}(\RFNQ_{>e}(\Lstr_{T}, l)$ is in $[l..r]$, it holds that $r' - l' + 1 \neq 0$ and $\LF_{T}(s) \in [l'..r']$.
    Note that $\LF_{T}(s)$ is not necessarily $l'$ because 
    p-encoded suffixes $\penc{T[\SUR_{T}(p)..]}$ with $\Lstr_{T}[p] > e$ in $[l..r]$
    do not necessarily preserve the lexicographic order when they are extended by one symbol to the left,
    making it non-straightforward to identify the position $l'$.

    To tackle this problem, we consider $[l_e..r_e] = \GetInterval(s, e)$ and $[l'_{e+1}..r'_{e+1}] = \GetInterval(\LF_{T}(s), e+1)$,
    and show that $l' = l'_{e+1} + x$, where $x$ is the number of positions $p$ such that $\Lstr_{T}[p] > e$ with $p \in [l_e..l-1]$.
    Observe that $[l..r] \subseteq [l_e..r_e]$ and $[l'..r'] \subseteq [l'_{e+1}..r'_{e+1}]$ by definition, and
    that $\LF_{T}(p) \in [l'_{e+1}..r'_{e+1}]$ if and only if $p \in [l_{e}..r_{e}]$ and $\Lstr_{T}[p] > e$ (see Fig.~\ref{fig:bws} for an illustration).
    Also, it holds that $\penc{T[\SUR_{T}(\LF_{T}(p))..]} < \penc{T[\SUR_{T}(\LF_{T}(q))..]}$
    for any $p \in [l_{e}..l-1]$ and $q \in [l..r]$ with $\Lstr_{T}[p] > e$ and $\Lstr_{T}[q] > e$
    because $\lcpinf(\penc{T[\SUR_{T}(p)..]} < \penc{T[\SUR_{T}(q)..]}) = e$, and they fall into Case (B4) of Lemma~\ref{lem:penc_order}.
    Similarly for any $p \in [l..r]$ and $q \in [r+1..r_{e}]$ with $\Lstr_{T}[p] > e$ and $\Lstr_{T}[q] > e$, 
    we have $\penc{T[\SUR_{T}(\LF_{T}(p))..]} < \penc{T[\SUR_{T}(\LF_{T}(q))..]}$.
    Hence, $l' = l'_{e+1} + x$ holds.
\end{itemize}
  This concludes the proof.
\end{proof}

\begin{figure}[t]
  \center{%
    \includegraphics[scale=0.33]{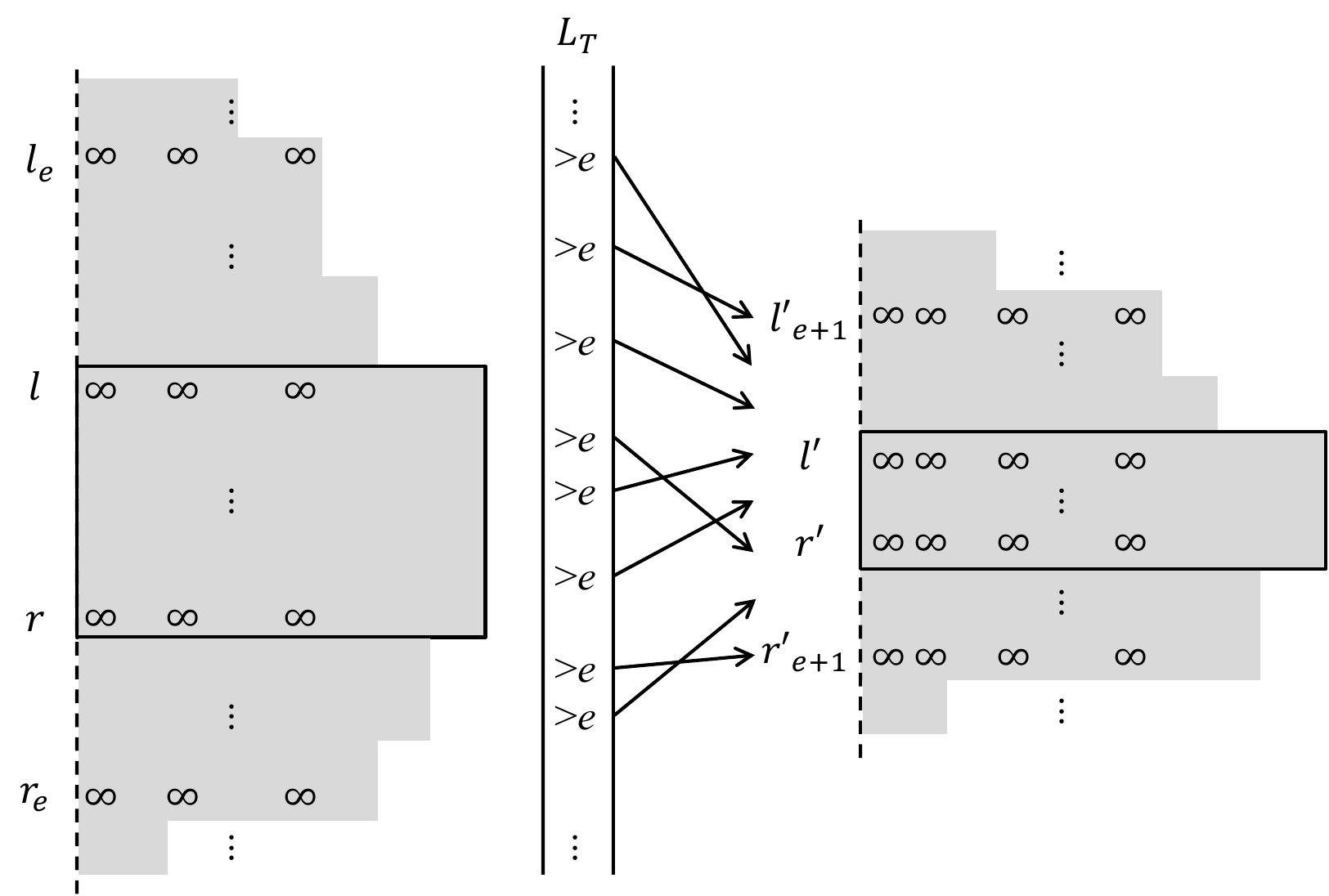}
  }
  \caption{Illustration for the computation of $cw$-interval $[l'..r']$ from $w$-interval $[l..r]$ 
    for the case when $c$ is a p-symbol that does not appear in $w$ with $e = |w|_{p} = 3$.
    The left (resp. right) part shows sorted p-encoded suffixes around $[l..r]$ (resp. $[l'..r']$) with
    grayed areas representing the longest common prefix between $w$ (resp. $cw$) and each p-encoded suffix.
    It depicts that the positions $p$ in $[l_e..l-1]$ with $\Lstr_{T}[p] > e$ are mapped to $[l'_{e+1}..l'-1]$ by LF-mapping.
  }
  \label{fig:bws}
\end{figure}

We are now ready to prove the main theorem:
\begin{proof}[Proof of Theorem~\ref{theo:index}:]
  If we only need counting queries, Lemmas~\ref{lemma:online} and~\ref{lemma:bws} are enough:
  While we build $\Lstr_{T}$, $\Fstr_{T}$ and $\LCPinf_{T}$ online, 
  we can compute $w$-interval $[l..r]$ for a given pattern $w$ of length $m$ using Lemma~\ref{lemma:bws} successively $m$ times,
  spending $O(m \frac{\lg \psigma \lg n}{\lg \lg n})$ time in total.

  Since $\{ \SUR_{T}(i) \mid i \in [l..r] \}$ is the set of occurrences of $w$ in $T$,
  we consider how to access $\SUR_{T}(i)$ in $O(\frac{\lg^2 n}{\lg \sigma \lg \lg n})$ time to support locating queries.
  As is common in BWT-based indexes, we employ a sampling technique (e.g., see~\cite{Ferragina2000ODS}):
  For every $\Theta(\log_{\sigma} n)$ text positions we store the values so that 
  if we apply LF/FL-mapping to $i$ successively at most $\Theta(\log_{\sigma} n)$ times we hit one of the sampled text positions.
  A minor remark is that since our online construction proceeds from right to left, 
  it is convenient to start sampling from the right-end of $T$
  and store the distance to the right-end instead of the text position counted from the left-end of $T$.

  During the online construction of the data structures for $\Lstr_{T}$, $\Fstr_{T}$ and $\LCPinf_{T}$,
  we additionally maintain a dynamic bit vector of length $n$ and dynamic integer string $V_{T}$ of length $O(n / \log_{\sigma} n)$,
  which marks the sampled positions and stores sampled values, respectively.
  We implement $V_{T}$ with the dynamic string of Lemma~\ref{lem:drs} in $O(n \frac{\lg n}{\log_{\sigma} n }) = O(n \lg \sigma)$ bits with $O(\lg n)$ query times.
  In order to support LF/FL-mapping in $O(\frac{\lg n}{\lg \lg n})$ time,
  we also maintain a dynamic string of Lemma~\ref{lem:drs} for $\Lstr_{T}$.
  With the additional space usage of $O(n \lg \sigma)$ bits,
  we can access $\SUR_{T}(i)$ in $O(\frac{\lg^2 n}{\lg \sigma \lg \lg n})$ time
  as we use LF/FL-mapping at most $\Theta(\log_{\sigma} n)$ times.
  This leads to the claimed time bound for locating queries.
\end{proof}

\section*{Acknowledgements}
This work was supported by JSPS KAKENHI Grant Number 19K20213 and 22K11907 (TI).

\bibliography{refs}

\end{document}